\documentclass{article}

\usepackage{amsmath,amsfonts,amssymb,amsthm}
\usepackage{setspace,graphicx}
\usepackage{natbib}
\usepackage{bm}
\usepackage{ifthen}
\usepackage{verbatim}
\usepackage{color}
\usepackage{dsfont}
\usepackage{tikz}
\usepackage{setspace,graphicx}
\usepackage{sgame}
\usepackage{bm}
\usepackage{multirow}
\usepackage{ifthen}
\usepackage{verbatim}
\usepackage{xcolor}
\usetikzlibrary{arrows.meta}

\theoremstyle{plain}
\theoremstyle{definition}
  \newtheorem{theorem}{Theorem}

  \newtheorem{corollary}{Corollary}
  \newtheorem{definition}{Definition}

  \newtheorem{proposition}{Proposition}
  \newtheorem{remark}{Remark}
  \theoremstyle{remark}

\def\RR{\mathbb{R}}
\DeclareMathOperator*{\argmax}{arg\,max}

\begin{document}


\title{Success functions in large contests}

\author{Yaron Azrieli\thanks{Department of Economics, The Ohio State University, 1945 North High street, Columbus, OH 43210, azrieli.2@osu.edu.}~ and Christopher P. Chambers\thanks{Department of Economics, Georgetown University, ICC 580  37th and O Streets NW, Washington DC 20057, cc1950@georgetown.edu}}

\maketitle

\begin{abstract}
We consider contests with a large set (continuum) of participants and axiomatize contest success functions that arise when performance is composed of both effort and a random element, and when winners are those whose performance exceeds a cutoff determined by a market clearing condition. A co-monotonicity property is essentially all that is needed for a representation in the general case, but significantly stronger conditions must hold to obtain an additive structure. We illustrate the usefulness of this framework by revisiting some of the classic questions in the contests literature. 




\end{abstract}

\bigskip
\bigskip
\bigskip
\bigskip
\bigskip
\bigskip
\bigskip
\bigskip
\bigskip
\bigskip
\bigskip
\bigskip

 \begin{quote}
``The idea of a continuum of traders may seem outlandish to the reader. Actually, it is no stranger than a continuum of prices or of strategies or a continuum of ``particles'' in fluid mechanics\ldots ~It should be emphasized that our consideration of a continuum of traders is not merely a mathematical exercise; it is the expression of an economic idea.''

\hfill In: \emph{Markets with a continuum of traders}, Robert J. Aumann, Econometrica 1964
 \end{quote}

\newpage
\section{Introduction}\label{sec-introduction}

Models of contests, in which agents compete for awards by making irrevocable investments, have been used in the literature to analyze a wide variety of phenomena. Starting from the seminal works of \citet{tullock1980efficient, tullock2001efficient} and \citet{lazear1981rank}, many of these models assume that the identity of the winner(s) is not a deterministic function of agents' choices, i.e., that the outcome has some randomness to it. The mapping from strategy profiles to the probability of winning of each participant is known as a \emph{Contest Success Function (CSF)}; it is a key ingredient of the model determining agents' incentives and consequentially the equilibria of the game.

We study CSFs when the set of participants is large (a continuum). As demonstrated in the literature (e.g., \citet{olszewski2016large, olszewski2019bid, olszewski2020performance}, \citet{morgan2018ponds}, \citet{adda2023grantmaking},  \citet{azrieli2024temporary}), as well as in the last section of this paper, working ``in the limit'' with an infinite set of contestants is often more tractable than with a large but finite set. However, specifying a CSF in this context is not as straightforward as in the finite case, and the properties of such functions are not well understood. The main contribution of this work is to axiomatize a class of CSFs that we refer to as \emph{Random Performance Functions (RPFs)}. Roughly speaking, in an RPF the effort choice of an agent is randomly mapped into performance, and the agent wins if their performance exceeds a cutoff set by the principal; the cutoff is adjusted based on the distribution of efforts in the population to satisfy a market clearing condition, namely, that the total mass of winners is equal to the budget.\footnote{\citet{morgan2018ponds}, \citet{adda2023grantmaking} and  \citet{azrieli2024temporary} all use RPFs in their models.}

The first main result characterizes RPFs without assuming any particular structure of the mapping from effort to performance; the only restriction is that higher effort implies higher distribution of performance in the sense of first-order stochastic dominance. We show that the key property needed for such a representation is Co-monotonicity, which means that if an agent exerting effort $e$ has a higher winning probability when facing competition (distribution of efforts) $p$ than when facing competition $p'$, then the same is true for any other effort level $e'$. In other words, under Co-monotonicity one can rank effort distributions $p$ based on how competitive they are in an unambiguous way. Theorem \ref{thm-characterization} shows that this axiom, along with standard continuity and monotonicity properties, are necessary and sufficient for a CSF to be an RPF. 

In Theorem \ref{thm-A-RPF} we add more structure and characterize the popular specification where performance is the sum of effort plus a noise term. As can be expected, this is much more demanding as it requires invariance of the winning probability to certain shifts of the effort distribution. We further show that in this case the distribution of noise can be identified from the CSF up to translations.

We view these two theorems as providing a foundation for the use of RPFs in models of large contests. Isolating the necessary and sufficient conditions clarifies the structure of agents' incentives under this type of CSF. Alternatively, in environments where the CSF can be estimated from observable data our results facilitate testing whether an RPF is an appropriate modeling choice. As an example, consider competition among scientists for research grants: The quality distribution of submissions in a given cycle can be estimated ex-post by the impact of each proposal, and so does the rate of success of each quality level; one can then test whether Co-monotonicity holds to determine if an RPF captures the evaluation process of proposals.


In the finite-agent case, the CSF proposed by \citet{tullock1980efficient} and its generalizations are by far the most popular. \citet{skaperdas1996contest} provides the first axiomatization of this class by using a population consistency condition. \citet{clark1998contest} provide an asymmetric generalization of the homogeneous special case of \citet{skaperdas1996contest}, interpreting the consistency condition as an analogue of the choice axiom of \citet{luce1959individual}. Other generalizations include \citet{myerson2006population} for the case of unknown population size, \citet{munster2009group} who considers contests between groups, and \citet{rai2009generalized} where agents' investments are multidimensional. \citet{berry1993rent} was the first to propose an extension of Tullock's CSF to the case of multiple prizes as we have here; a critique of Berry's extension and an alternative specification were suggested by \citet{clark1996multi}. \citet{fu2012micro} elaborate on the connection between noisy performance tournaments as in \citet{lazear1981rank} and the model of \citet{clark1996multi}. See \citet{sisak2009multiple} for a survey of the literature on multi-prize contests.

Recently there have been several attempts to extend Tullock-type CSFs to contests with a continuum of agents, see for example \citet{lahkar2023optimal, lahkar2023rent} and \citet{dougan2023large}. The goal of these papers is to obtain explicit characterization of the equilibrium in cases where this is not possible with a finite set of contestants. However, the range of these proposed CSFs is not bounded above by one, so it is hard to interpret them as winning probabilities or shares of a prize. The RPFs axiomatized in this paper do not have this issue. While we do not obtain closed form expressions for equilibria, we illustrate in Section \ref{sec-applications} that they can be used to address some of the classic questions in the contests literature.

\section{Contest success functions}\label{sec-CSF}

There is a continuum of agents with total mass one that can choose whether to participate in a contest and if so how much effort to exert. The contest designer's budget is sufficient to give prizes to a mass of $k\in (0,1)$ agents. We refer to $k$ as the budget and fix it until Section \ref{sec-comparison}.\footnote{We assume that all winners receive the same prize so that there are only two possible outcomes -- win or lose -- for each agent.} 

Let $E=(0,\infty)$ be the set of possible effort levels agents can choose from. A \emph{measure on $E$} is a non-negative, $\sigma$-additive function defined over the Borel sets in $E$. We denote by $\Delta_k$ the set of measures $p$ on $E$ with $p(E)\in (k,1]$. Each $p\in \Delta_k$ is interpreted as a potential `distribution of efforts' corresponding to some strategy profile. We sometimes refer to $p$ as the \emph{competition}. We allow for $p(E)<1$ to capture situations where not all agents choose to participate in the contest, namely, $p(E)$ is the fraction who participate and $1-p(E)$ stay out. The restriction to $p(E)>k$ guarantees that there is not enough budget to reward all participants; if $p(E)\le k$ we simply assume that any effort level guarantees a win.

It will be convenient to use the following notation. For any $e\in E$ the Dirac measure $\delta_e\in \Delta_k$ is defined by $\delta_e(A)=1$ if $e\in A$ and $\delta_e(A)=0$ otherwise. If $p\in \Delta_k$ and $\bar e\in E$ then $p\dotplus \bar e$ is the `right-shift' of $p$ by $\bar e$, namely, the measure defined by $(p\dotplus \bar e)(A)= p (\{e-\bar e: e\in A\cap (\bar e,+\infty)\})$. Finally, we write $\int$ rather than $\int_E$ whenever the domain of integration is $E$.

The object of interest is a mapping $W:E\times \Delta_k \to [0,1]$. Here $W(e,p)$ is interpreted as the probability of winning for an agent that exerts effort $e$ given that they face competition $p$.\footnote{Note that $W(e,p)$ is defined even if $e$ is not in the support of $p$.  This allows one to determine whether a given strategy profile is a Nash equilibrium by considering all possible effort deviations.} Following the literature, we refer to $W$ as a \emph{Contest Success Function}. 

\begin{definition}\label{def-CSF}
    A Contest Success Function (CSF) is a mapping $W:E\times \Delta_k \to [0,1]$ such that, for any fixed $p\in \Delta_k$, $W(e,p)$ is measurable in $e$ and satisfies 
    \begin{equation}\label{eqn-BB}
    \int W(e,p)dp(e) = k.
\end{equation}
\end{definition}

Equation (\ref{eqn-BB}) is a market clearing condition requiring that for any $p$ the total mass of winners is equal to $k$. In other words, the designer spends their entire budget irrespective of the distribution of efforts.\footnote{This is similar to the constraint that the sum of winning probabilities of all agents is 1 in the classic finite setup, see for example property A1 in \citet{skaperdas1996contest}.}

\section{Random Performance Functions}\label{sec-RPF}

We are interested in characterizing those CSFs $W$ that arise as a result of the following process: The performance of an agent that exerts effort $e$ is randomly drawn according to a distribution that depends on $e$, where higher effort is more likely to result in better performance; the designer sets a threshold performance value $s$ and the winners are those whose performances exceed this cutoff; the threshold $s$ is adjusted based on the competition $p$ so that (\ref{eqn-BB}) is satisfied. This is formalized in the next definition.    

\begin{definition}\label{def-RPF}
    A CSF $W$ is a \emph{Random Performance Function (RPF)} if there exists a collection of cumulative distribution functions $\{F_e\}_{e\in E}$ such that the following hold:
    \begin{enumerate}
        \item Each $F_e$ is continuous and strictly increasing.
        
        \item For each $x\in \RR$, the mapping $e\mapsto F_e(x)$ is continuous, strictly decreasing, and satisfies $\lim_{e\uparrow +\infty} F_e(x)=0$.

        \item For each $e\in E$ and $p\in \Delta_k$ it holds that $W(e,p)=1-F_e(s(p))$, where $s(p)\in \RR$ is defined as the unique solution to the equation $\int [1-F_e(s(p))]dp(e)=k$.  
    \end{enumerate}
\end{definition}

The first condition in the definition is for technical reasons. The second implies that the distribution of performance for higher effort levels first-order stochastically dominates the distribution for lower effort levels, and that very high effort is likely to result in high performance. Condition 3 means that $W$ can be represented by the collection $\{F_e\}_{e\in E}$. 

\begin{remark}\label{remark-s(p)}
    Existence and uniqueness of $s(p)$ that satisfies the equality in the above definition is guaranteed by our assumptions. Indeed, by property 2 of the definition $1-F_e(s)$ is a measurable function of $e$ so the integral $\int [1-F_e(s)]dp(e)$ is well-defined, and this integral is continuous and strictly increasing in $s$ by part 1 of the definition. Further, this integral converges to zero as $s\uparrow +\infty$ and converges to $p(E)>k$ as $s\downarrow -\infty$. Thus, there exists a unique $s$ such that $\int [1-F_e(s)]dp(e)=k$ holds.
\end{remark}

\subsection{Characterization}

To characterize RPFs we introduce the following properties. 

\bigskip

\noindent \textbf{$e$-Continuity:} For any fixed $p\in \Delta_k$, $W(e,p)$ is continuous in $e$.

\medskip

\noindent \textbf{$p$-Continuity:} For any fixed $e\in E$ and $p,p'\in \Delta_k$, $W(e,\alpha p+ (1-\alpha)p')$ is a continuous function of $\alpha\in [0,1]$.

\medskip

\noindent \textbf{Monotonicity:} For any fixed $p\in \Delta_k$, $W(e,p)$ is strictly increasing in $e$ and $\lim_{e\uparrow +\infty} W(e,p)=1$. 

\medskip

\noindent \textbf{Competitiveness:} For any fixed $e\in E$, $\lim_{\bar e \uparrow +\infty} W(e,\delta_{\bar e})=0$.

\medskip

\noindent \textbf{Co-monotonicity:} For all $e\in E$ and $p,p'\in \Delta_k$, if $W(e,p)\ge W(e,p')$ then $W(e',p)\ge W(e',p')$ for all $e'\in E$.

\begin{theorem}\label{thm-characterization}
    The contest success function $W$ is an RPF if and only if it satisfies $e$-Continuity, $p$-Continuity, Monotonicity, Competitiveness, and Co-monotonicity. 
\end{theorem}

\begin{proof}
    (Necessity) Suppose that $W$ is an RPF. First, $e$-Continuity is immediate since $W(e,p)=1-F_e(s(p))$ and the mapping $e\mapsto F_e(x)$ is assumed to be continuous for any fixed $x$. To show $p$-Continuity it is enough to prove that $s(\alpha p+(1-\alpha)p')$ is a continuous function of $\alpha$ for any given pair $p,p'\in \Delta_k$. Consider a sequence $\{\alpha_n\}$ with $\alpha_n\to \alpha\in (0,1)$. For each $n$ denote $s_n=s(\alpha_n p+(1-\alpha_n)p')$. The sequence $\{s_n\}$ is clearly bounded (e.g. $s_n$ is between $s(p)$ and $s(p')$ for all $n$), so it has a convergent subsequence $s_{n_l}\to s_0$. Note that for any $n$ we have
    $$k = \int(1-F_e(s_n))d(\alpha_n p+(1-\alpha_n)p')(e) = \alpha_n \int(1-F_e(s_n))dp(e) +(1-\alpha_n)\int(1-F_e(s_n))dp'(e).$$
    Therefore, using the dominated convergence theorem, 
    \begin{eqnarray*}
        k &=& \lim_l \alpha_{n_l} \int(1-F_e(s_{n_l}))dp(e) +(1-\alpha_{n_l})\int(1-F_e(s_{n_l}))dp'(e) =\\
        & & \alpha \int(1-F_e(s_0))dp(e) +(1-\alpha)\int(1-F_e(s_0))dp'(e),
    \end{eqnarray*}
    and therefore $s(\alpha p+(1-\alpha)p')=s_0$. We have established that $\{s_n\}$ has a subsequence converging to $s(\alpha p+(1-\alpha)p')$, but note that we could have applied the same argument starting from any subsequence of $\{s_n\}$ rather than from the original sequence. Thus, any subsequence of $\{s_n\}$ has a sub-subsequence that converges to $s(\alpha p+(1-\alpha)p')$, and hence the entire sequence converges to this limit.

    To show Monotonicity, fix some $p$ and suppose that $e<e'$. Then $W(e,p)=1-F_e(s(p))<1-F_{e'}(s(p))=W(e',p)$, where the inequality is by condition 2 in Definition \ref{def-RPF}. Also, by the same condition we have that $\lim_{e\uparrow +\infty} W(e,p)=\lim_{e\uparrow +\infty} 1-F_e(s(p))=1$.

    Next, consider the function $g(\bar e)=s(\delta_{\bar e})$. Applying (\ref{eqn-BB}) to the measure $\delta_{\bar e}$ we have that $F_{\bar e}(g(\bar e)) =1-k$. Thus, by condition 2 in Definition \ref{def-RPF}, $g$ is an increasing function. If $g$ was bounded, say by $\bar s$, then we would have $F_{\bar e}(\bar s)\ge F_{\bar e}(g(\bar e)) =1-k$ for all $\bar e$, which is impossible since condition 2 in Definition \ref{def-RPF} requires that $\lim_{\bar e \uparrow +\infty} F_{\bar e}(\bar s)=0$. It follows that $\lim_{\bar e \uparrow +\infty} s(\delta_{\bar e})=+\infty$ and therefore that $\lim_{\bar e \uparrow +\infty} W(e,\delta_{\bar e})=0$ for any fixed $e$, proving Competitiveness. 
    
    Finally, Co-monotonicity is immediate as well since $W(e,p)\ge W(e,p')$ implies that $s(p)\le s(p')$, from which $W(e',p)\ge W(e',p')$ follows.

    \bigskip
     
    (Sufficiency) Suppose $W$ satisfies $e$-Continuity, $p$-Continuity, Monotonicity, Competitiveness and Co-monotonicity. Fix some $e_0\in E$. We break the proof into several steps.

\medskip

    \noindent \textbf{Step 1:} We claim that $\{W(e_0,p) ~:~ p\in \Delta_k\}=(0,1)$. To see this, denote by $A$ the set on the left-hand side of the equality. From Competitiveness we immediately get that $\inf A=0$. Consider $p=\alpha\delta_{e_0}$ for some $\alpha>k$. By (\ref{eqn-BB}) we have that $\alpha W(e_0,p)=k$. Taking the limit as $\alpha$ converges to $k$ shows that $\sup A=1$. Finally, note that if for some $p,p'$ we have $W(e_0,p)\ge W(e_0,p')$, then $W(e_0,p)\ge W(e_0,\alpha p+(1-\alpha)p') \ge W(e_0,p')$ for any $\alpha\in (0,1)$. Indeed, suppose that was not the case, say $W(e_0,p)< W(e_0,\alpha p+(1-\alpha)p')$. Then by Co-monotonicity the same inequality would be true for every $e$. But then 
    \begin{eqnarray*}
        k &=& \int W(e,\alpha p+(1-\alpha)p')d(\alpha p+(1-\alpha)p')(e) =\\ 
        & & \alpha \int W(e,\alpha p+(1-\alpha)p')dp(e) +(1-\alpha) \int W(e,\alpha p+(1-\alpha)p')dp'(e) >\\
        & & \alpha \int W(e,p)dp(e) +(1-\alpha) \int W(e,p')dp'(e)=k,
    \end{eqnarray*}
    where the first and last equalities are from (\ref{eqn-BB}). A similar argument applies if $W(e_0,p')< W(e_0,\alpha p+(1-\alpha)p')$. Combining this with $p$-Continuity proves that indeed $A=(0,1)$.

\medskip

    \noindent \textbf{Step 2:} Let $\phi:(0,1)\to \RR$ be a strictly increasing and continuous function whose image is the entire real line. For any $p\in \Delta_k$ define $s(p)=\phi(1-W(e_0,p))$. Notice that $s(p)\ge s(p')$ if and only if $W(e_0,p) \le W(e_0,p')$. For any $x\in \RR$ choose $p_x$ such that $s(p_x)=x$; existence of $p_x$ is guaranteed by Step 1. Define $F_e(x)=1-W(e,p_x)$ for every $e\in E$. 

\medskip

    \noindent \textbf{Step 3:} In this step we show that each $F_e$ just defined is a continuous and strictly increasing cdf. Fix $e$. For monotonicity, let $x<x'$ and note that this implies $s(p_x)<s(p_{x'})$ and therefore $W(e_0,p_x) > W(e_0,p_{x'})$. It follows from Co-monotonicity that $W(e,p_x) > W(e,p_{x'})$ and hence $F_e(x)<F_e(x')$. 
    
    Next, we claim that for any $p$ there is $x$ such that $W(e,p)=W(e,p_x)$. Indeed, by construction we have that $W(e_0,p)=W(e_0,p_{s(p)})$, so by Co-monotonicity $W(e,p)=W(e,p_{s(p)})$ as well. Therefore, 
    $$\{W(e,p_x) ~:~ x\in \RR\}= \{W(e,p) ~:~ p\in \Delta_k\}=(0,1),$$
    where the last equality is as in Step 1 of this proof. Together with the monotonicity property just shown this implies that $\lim_{x\downarrow-\infty}F_e(x)=0$ and $\lim_{x\uparrow +\infty}F_e(x)=1$. 
    
    Finally, for continuity, let $x_n \rightarrow x$. Then there are $\underline{x}$, $\overline{x}$ for which $\underline{x}<x_n<\overline{x}$ for all $n$. Using the same argument as in Step 1, for each $n$ we can find $\alpha_n\in (0,1)$ such that $s(\alpha_n p_{\underline{x}}+(1-\alpha_n)p_{\overline{x}})=x_n$. Take a subsequence $\alpha_{n_l}$ converging to some $\alpha$ and observe that 
    \begin{equation}\label{eqn-continuity}
        x=\lim x_{n_l} = \lim s(\alpha_{n_l} p_{\underline{x}}+(1-\alpha_{n_l})p_{\overline{x}}) = s(\alpha p_{\underline{x}} + (1-\alpha)p_{\overline{x}}),
    \end{equation}
    where the last equality follows from $p$-Continuity. Consequently, 
    $$F_e(x_{n_l})= 1 - W(e,p_{x_{n_l}}) =  1 - W(e,\alpha_{n_l} p_{\underline{x}}+(1-\alpha_{n_l})p_{\overline{x}}) \to 1-W(e,\alpha p_{\underline{x}}+(1-\alpha) p_{\overline{x}})=F_e(x),$$
    where the first equality is by definition, the next follows from $s(\alpha_n p_{\underline{x}}+(1-\alpha_n)p_{\overline{x}})=x_n=s(p_{x_n})$ and Co-monotonicity, the limit is by $p$-Continuity, and the last equality is by (\ref{eqn-continuity}). Since we could have applied the same argument starting from any subsequence of $\{x_n\}$, it follows that $F_e(x_n)\to F_e(x)$.

\medskip

    \noindent \textbf{Step 4:} Here we argue that condition 2 of Definition \ref{def-RPF} is satisfied. Fix $x\in \RR$. Since $F_e(x)=1-W(e,p_x)$, and by $e$-Continuity, we have that $e\mapsto F_e(x)$ is continuous. This mapping is strictly decreasing by Monotonicity. And $\lim_{e\uparrow +\infty} F_e(x)=0$ also directly follows from Monotonicity.

\medskip

    \noindent \textbf{Step 5:} We conclude by showing that the collection $\{F_e\}$ represents $W$ as in condition 3 of Definition \ref{def-RPF}. Fix some $\bar p\in \Delta_k$ and denote $\bar x=s(\bar p)$. Then by construction we have that $s(p_{\bar x})=\bar x=s(\bar p)$ and therefore $W(e_0,\bar p) = W(e_0,p_{\bar x})$. By Co-monotonicity, $W(e,\bar p) = W(e,p_{\bar x})=1-F_e(s(\bar p))$ for all $e\in E$, and it follows from (\ref{eqn-BB}) that    
    $$k = \int W(e,\bar p)d \bar p(e) = \int (1-F_e(s(\bar p))) d \bar p(e).$$
    
\end{proof}


\section{Additive noise}\label{sec-additive}

The definition of an RPF allows for very general relationship between effort and performance. Indeed, the only substantial requirement is that higher effort implies higher distribution of performance in the sense of first-order stochastic dominance. In this section we add more structure and characterize the case where performance is obtained by adding a random shock to the effort choice. More explicitly, performance is equal to $e+X$, where $X$ is a random variable representing the noise term. 

\begin{definition}\label{def-additive}
   The contest success function $W$ is an Additive RPF (A-RPF) if there is a continuous and strictly increasing cdf $F$ such that $W$ has an RPF representation with $F_e(x)=F(x-e)$ for every $e\in E$.
\end{definition}

\subsection{Characterization}

To characterize Additive RPFs we introduce the following two additional properties. 

\bigskip

\noindent \textbf{Invariance to Common Shifts:}  For every $e,a\in E$ and $p\in \Delta_k$, $W(e,p)=W(e+a,p\dotplus a)$.

\medskip

\noindent \textbf{Invariance to $p$ Shifts:}  For every $e,e',a\in E$ and $p,p'\in \Delta_k$, if $W(e,p)=W(e',p')$, then $W(e,p\dotplus a)=W(e',p' \dotplus a)$.

\begin{theorem}\label{thm-A-RPF}
    The contest success function $W$ is an A-RPF if and only if it satisfies $e$-Continuity, Monotonicity, Invariance to Common Shifts, and Invariance to $p$ Shifts.
\end{theorem}

\begin{proof}
(Necessity) Suppose $W$ is an A-RPF with cdf $F$. Since $W$ is in particular an RPF, $e$-Continuity and Monotonicity follow from Theorem \ref{thm-characterization}. Fix some $p\in \Delta_k$ and $a\in E$. Then
$$\int [1-F(s(p)+a-e)]d(p\dotplus a)(e) =  \int [1-F(s(p)-e)]d(p)(e) = \int W(e,p)dp(e)=k,$$
where the first equality is a change of variable, the second by the assumption that $F$ represents $W$, and the last by (\ref{eqn-BB}). It follows that $s(p\dotplus a)=s(p)+a$. 

Now, for Invariance to Common Shifts, we have for every $e,a\in E$ and $p\in \Delta_k$,
$$W(e+a,p\dotplus a) = 1-F(s(p\dotplus a)-(e+a)) = 1-F(s(p)+a-(e+a)) = 1-F(s(p)-e)=W(e,p).$$
Similarly, suppose that $W(e,p)=W(e',p')$ holds for some $e,e'\in E$ and $p,p'\in \Delta_k$. Then $1-F(s(p)-e)=1-F(s(p')-e')$, which, by the strict monotonicity of $F$, implies that $s(p)-e=s(p')-e'$. But then for any $a\in E$,
$$s(p\dotplus a)-e = s(p)+a-e = s(p')+a-e' = s(p'\dotplus a)-e'.$$
This shows that $W(e,p\dotplus a)=W(e',p' \dotplus a)$ and hence that $W$ satisfies Invariance to $p$ Shifts.

\bigskip

(Sufficiency) Fix some arbitrary $(e^*,p^*)\in E\times \Delta_k$. Define $F$ by
\begin{eqnarray}\label{eqn-F}
F(x)=
\left\{ \begin{array}{ll}
 1-W(e^*,p^*\dotplus x) & \textit{ if } ~~ x\ge 0,\\
1-W(e^*-x,p^*) & \textit{ if } ~~ x<0.\\
\end{array} \right.
\end{eqnarray}

We start by showing that $F$ is strictly increasing. If $x_1<x_2\le 0$ then by Monotonicity $W(e^*-x_1,p^*)>W(e^*-x_2,p^*)$, so $F(x_1)<F(x_2)$. If $0\le x_1<x_2$ then
$$W(e^*,p^*\dotplus x_1) = W(e^*+(x_2-x_1),(p^*\dotplus x_1) \dotplus (x_2-x_1))) = W(e^*+(x_2-x_1),p^*\dotplus x_2) >  W(e^*,p^*\dotplus x_2),$$
where the first equality is by Invariance to Common Shifts, the next is just a simplification, and the inequality is by Monotonicity. Thus, $F(x_1)<F(x_2)$ and $F$ is strictly increasing.

Next, we argue that $F$ is continuous. For $x<0$ this follows immediately from $e$-Continuity. For $x\geq 0$, fix some $\bar e>x$. Then by Invariance to Common Shifts, 
$$W(e^*,p^*\dotplus x)=W(e^*+(\bar e-x),p^*\dotplus(x +(\bar e -x ) ))=W(e^*+\bar e-x,p^*\dotplus \bar e).$$
Thus, by $e$-Continuity we get that $F$ is continuous at any $x\geq 0$ as well.

It remains to show that $F$ represents $W$ as in Definition \ref{def-additive}. Fix some $p\in \Delta_k$. Our goal is to show that there is $s\in \RR$ such that $W(e,p)=1-F(s-e)$ for every $e\in E$. Indeed, if that is true then we would have 
$$\int [1-F(s-e)]dp(e) = \int W(e,p)dp(e) =k,$$
implying that $s$ is the performance cutoff for $p$ and therefore that $F$ represents $W$. 

Take some $e>e'$ two effort levels. By Monotonicity $W(e,p)\in (0,1)$, so there is a unique $s$ such that $W(e,p)=1-F(s-e)$. We need to show that $W(e',p)=1-F(s-e')$ and we distinguish between two cases.

\medskip

\underline{Case 1 -- $s\ge e$}: Here we have $W(e,p)=1-F(s-e) = W(e^*,p^*\dotplus (s-e))$. Thus, applying the axiom Invariance to $p$ Shifts with $a=e-e'$ we obtain
\begin{equation}\label{eqn-case1}
    W(e,p\dotplus(e-e')) = W(e^*, p^*\dotplus (s-e+e-e')) = W(e^*, p^*\dotplus (s-e')).
\end{equation}
We can now write
$$W(e',p) = W(e'+(e-e'),p\dotplus(e-e')) = W(e,p\dotplus(e-e') = W(e^*, p^*\dotplus (s-e')) = 1-F(s-e'),$$
where the first equality is by Invariance to Common Shifts, the next is a simplification, the next is from (\ref{eqn-case1}), and the last is from the definition of $F$ and using the assumption that $s\ge e>e'$. 

\medskip

\underline{Case 2 -- $e> s$}: In this case $W(e,p)=1-F(s-e) = W(e^*+e-s,p^*)$, so applying Invariance to $p$ Shifts with $a=e-e'$ gives
\begin{equation}\label{eqn-case2}
    W(e,p\dotplus(e-e')) = W(e^*+e-s, p^*\dotplus (e-e')).
\end{equation}
Let $b>0$ be large enough so that $s-e+b>0$. Then
\begin{eqnarray*}
    W(e',p) = W(e,p\dotplus(e-e')) &=& W(e^*+e-s, p^*\dotplus (e-e')) =\\
    & & W(e^*+e-s+(s-e+b), p^*\dotplus (e-e'+s-e+b)) =\\
    & & W(e^*+b, p^*\dotplus (s-e'+b)),
\end{eqnarray*}
where the first equality is by Invariance to Common Shifts, the next is by (\ref{eqn-case2}), the next is again by Invariance to Common Shifts, and the last is just a simplification. 

Now, if $s\ge e'$ then this last expression is equal to $W(e^*, p^*\dotplus (s-e'))$ by Invariance to Common Shifts, so we obtain 
$W(e',p)=W(e^*, p^*\dotplus (s-e')) = 1-F(s-e')$ as needed. If on the other hand $s< e'$ then again by Invariance to Common Shifts the last expression is equal to $W(e^*+e'-s, p^*)$ and hence $W(e',p)=W(e^*+e'-s, p^*) = 1-F(s-e')$. This completes the proof.
\end{proof}

\subsection{Uniqueness of the representation}

The following result states that if $W$ is an Additive RPF then the noise distribution is identified up to translations. 

\begin{proposition}
    Suppose that $W$ is an Additive RPF that can be represented by the cdf $F_1$, i.e., $W(e,p)=1-F_1(s_1(p)-e)$ for all $e,p$. Then $F_2$ also represents $W$ if and only if there is $t\in \RR$ such that $F_2(x)=F_1(x+t)$ for all $x\in \RR$. 
\end{proposition}

\begin{proof}
    (If): Suppose $F_1$ represents $W$ and fix $t$. Define $F_2(x)=F_1(x+t)$. Then for any $p$ we have
    $$\int [1-F_2(s_1(p)-t-e)]dp(e) = \int [1-F_1(s_1(p)-e)]dp(e) = k,$$
    where the first equality is by the definition of $F_2$ and the second from the characterizing property of $s_1(p)$. It follows that $s_2(p)=s_1(p)-t$. Thus, for all $e$ and $p$,
    $$W(e,p) = 1-F_1(s_1(p)-e) = 1-F_1(s_2(p)+t-e) = 1-F_2(s_2(p)-e)$$
    as claimed.

\medskip

    (Only if): Suppose both $F_1$ and $F_2$ represent $W$. Thus, $F_1(s_1(p)-e)=F_2(s_2(p)-e)$ for all $e$ and $p$. Fix some $\bar e>0$ and recall that $\delta_{\bar e}$ is the Dirac measure on $\bar e$. Then $1-F_1(s_1(\delta_{\bar e})-\bar e) = k$ by condition (\ref{eqn-BB}), so that $s_1(\delta_{\bar e})=F_1^{-1}(1-k)+\bar e$. By the same argument, $s_2(\delta_{\bar e})=F_2^{-1}(1-k)+\bar e$. Therefore, denoting $t=F_1^{-1}(1-k)-F_2^{-1}(1-k)$, we get
    $s_1(\delta_{\bar e}) - s_2(\delta_{\bar e})=t$. It follows that for every $\bar e$ and $e$
    $$F_2(s_2(\delta_{\bar e})-e) = F_1(s_1(\delta_{\bar e})-e) = F_1(s_2(\delta_{\bar e})+t-e),$$
    where the first equality is by assumption and the second by the definition of $t$ (note that $t$ is independent of $\bar e$). 

    To complete the proof we need to show that for any $x\in \RR$ we can find $\bar e,e>0$ such that $x=s_2(\delta_{\bar e})-e$, or equivalently, $x=F_2^{-1}(1-k)+\bar e -e$. This can be achieved by taking some $\bar e > x-F_2^{-1}(1-k)$ and letting $e=\bar e-(x-F_2^{-1}(1-k))$.
\end{proof}




\section{Comparing RPFs}\label{sec-comparison}

In this section we compare RPFs in terms of how well the probability of winning reflects effort rather than noise. Clearly, this is a key determinant of agents' incentives to exert effort in the contest. We start with the following definition.

\begin{definition}\label{def-compare}
    Consider two collections of cdfs $\mathcal{F}=\{F_e\}_{e\in E}$ and $\tilde{\mathcal{F}}=\{\tilde F_e\}_{e\in E}$ as in Definition \ref{def-RPF}. Say that $\mathcal{F}$ \emph{better reflects effort} than $\tilde{\mathcal{F}}$ if for every $0<k<1$, every $p\in \Delta_k$, and every $a\in E$ it holds that
\begin{equation}\label{eqn-compare}
    \int_a^\infty W(e,p)dp(e) \ge \int_a^\infty \tilde W(e,p)dp(e),
\end{equation}
    where $W$ and $\tilde W$ are the RPFs generated by $\mathcal{F}$ and $\tilde{\mathcal{F}}$, respectively.\footnote{Note that $W$ and $\tilde W$ depend also on the budget $k$; our definition requires the inequality to hold for all $k$. Note also that $\int_E W(e,p)dp(e) =k= \int_E \tilde W(e,p)dp(e)$, so we can use standard notions of stochastic dominance to make comparisons between $W$ and $\tilde W$.}
\end{definition}

In words, $\mathcal{F}$ better reflects effort than $\tilde{\mathcal{F}}$ if for every effort distribution chosen by the agents, the effort distribution among winners based on $\mathcal{F}$ first-order stochastically dominates the effort distribution among winners based on $\tilde{\mathcal{F}}$. When there is very little noise and each $F_e$ is concentrated around $e$, the winners will be those who put in most effort; as noise increases, more low-effort agents will get lucky and the effort distribution among winners will stochastically decrease. 

It turns out that the notion of better reflecting effort is closely related to an informativeness criterion due to \citet{lehmann1988comparing}, who was interested in comparing statistical experiments for determining the location of a distribution. The connection between large contests and Lehmann's condition has already been pointed out by \citet{adda2023grantmaking}. They showed that an increase in Lehmann's informativeness implies that the equilibrium entry cutoff increases, reducing participation by low quality types. The next proposition shows that Lehmann's informativeness precisely characterizes the noisiness of the contest's outcome.  

\begin{proposition}\label{prop-comparison}
    The collection $\mathcal{F}$ better reflects effort than the collection $\tilde{\mathcal{F}}$ if and only if the function $F^{-1}_e(\tilde F_e(x))$ is non-decreasing in $e$ for any fixed $x\in \RR$.
\end{proposition}

\begin{proof}
    Assume first that the condition in the proposition holds and fix some $0<k<1$ and $p\in \Delta_k$. Suppose that for some $e\in E$ we have $W(e,p)\ge \tilde W(e,p)$. Denoting by $s(p)$ and $\tilde s(p)$ the cutoffs corresponding to $\mathcal{F}$ and $\tilde{\mathcal{F}}$, respectively, we thus have that $F_{e}(s(p)) \le \tilde F_{e}(\tilde s(p))$, or equivalently $s(p)\le F^{-1}_{e}(\tilde F_{e}(\tilde s(p)))$. Thus, for every $e'>e$, 
    $$F_{e'}[s(p)] \le F_{e'}[F^{-1}_{e}(\tilde F_{e}(\tilde s(p)))] \le F_{e'}[F^{-1}_{e'}(\tilde F_{e'}(\tilde s(p)))] = \tilde F_{e'}[\tilde s(p)],$$
    where the first inequality is by the monotonicity of $F_{e'}$, and the second follows from the condition of the proposition applied for $x=\tilde s(p)$ and from the monotonicity of $F_{e'}$. It follows that $W(e',p)\ge \tilde W(e',p)$ for all $e'>e$. 

    By an analogous argument one obtains that if $W(e,p)\le \tilde W(e,p)$ then the same is true for all $e'<e$. Thus, $W$ and $\tilde W$ satisfy a single-crossing property. Since the integrals satisfy $\int W(e,p)dp(e) =k= \int \tilde W(e,p)dp(e)$, the inequality in (\ref{eqn-compare}) must hold for all $a\in E$.  

    Conversely, fix $k\in (0,1)$ and consider the effort distribution $p=\mu \delta_e + (1-\mu) \delta_{e'}$ for some $e<e'$ and some $\mu\in (0,1)$. Then from (\ref{eqn-compare}) we must have that $W(e',p)\ge \tilde W(e',p)$, and since both integrals over the entire domain $E$ coincide we also have $W(e,p)\le \tilde W(e,p)$. Thus, $F_{e'}(s(p)) \le \tilde F_{e'}(\tilde s(p))$ and $F_{e}(s(p)) \ge \tilde F_{e}(\tilde s(p))$. It follows that
    $$F^{-1}_e[\tilde F_e(\tilde s(p))] \le F^{-1}_e[F_e(s(p))] = s(p) \le F^{-1}_{e'}[\tilde F_{e'}(\tilde s(p))].$$
    Now, since $\mu \tilde F_e+(1-\mu)\tilde F_{e'}$ is a continuous cdf, for every $x\in \RR$ there is some $k\in (0,1)$ such that $\tilde s(p)=x$. This proves the condition in the proposition.
\end{proof}

As can be seen in the above proof, under Lehmann's condition not only that the induced CSF $W$ (first-order) stochastically dominates $\tilde{W}$, but the two satisfy a stronger single-crossing property. The connection between Lehmann's informativeness and single-crossing is well-known in the context of statistical decision theory, see for example \citet{quah2009comparative} and \citet{athey2018value}.

\citet{lehmann1988comparing} provides equivalent conditions for his informativeness ranking that can be applied in the special case of Additive RPFs. Translated to our setup it yields the following corollary.
 
\begin{corollary}\label{coro-compare}
\cite[Theorem 5.2]{lehmann1988comparing} Let $F,\tilde F$ be strictly increasing cdfs with corresponding densities $f,\tilde f$ such that $\log(f)$ and $\log(\tilde{f})$ are concave. Consider the collections of cdfs $\mathcal{F}=\{F(x-e)\}_{e\in E}$ and $\tilde{\mathcal{F}}=\{\tilde{F}(x-e)\}_{e\in E}$. Then $\mathcal{F}$ better reflects effort than $\tilde{\mathcal{F}}$ if and only if
$F^{-1}(q)- \tilde{F}^{-1}(q)$ is weakly decreasing in $q\in (0,1)$. In particular, if $F^{-1}$ and $\tilde{F}^{-1}$ are differentiable then $\mathcal{F}$ better reflects effort than $\tilde{\mathcal{F}}$ if and only if $f[F^{-1}(q)]\ge \tilde{f}[\tilde{F}^{-1}(q)]$ for all $q\in(0,1)$.
\end{corollary}


\section{Applications}\label{sec-applications}

In this final section we illustrate the usefulness of the RPF framework by revisiting some of the classic questions in the contests literature. Our goal is not to present a comprehensive study of these issues, on which dozens of papers have been written, but to demonstrate that this model is tractable and may in fact simplify the analysis relative to models with a finite number of agents.

\subsection{Equilibrium with identical agents}

Consider a unit mass of agents, each of which chooses whether to participate in a contest, and if so how much effort $e\in E$ to exert. Staying out costs zero, while exerting effort $e$ has a cost $c(e)$, where $c:E\to \RR_+$ is twice differentiable and satisfies the standard conditions $\lim_{e\downarrow 0} c(e)=\lim_{e\downarrow 0} c'(e)=0$, $c'>0$, and $c''>0$. An agent who wins gets a prize of $V>0$. The total utility of an agent that exerts effort $e$ is $u(V)-c(e)$ if they win and $u(0)-c(e)$ if they do not, where the utility function $u$ is strictly increasing, (weakly) concave, and satisfies $u(0)=0$.   

Instead of explicitly considering strategy profiles, we focus on the population distribution of efforts they induce; let $W(e,p)$ be the probability of winning for an agent who exerts effort $e\in E$ given $p\in \Delta_k$.\footnote{If $p(E)\le k$ then set $W(e,p)=1$ for every $e\in E$.} The expected payoff of such agent is
$$U(e,p):=W(e,p)u(V)-c(e).$$

Say that $\bar e$ is a best response to $p$ if $\bar e\in \argmax_eU(e,p)$, and that $p$ is an equilibrium distribution (an equilibrium, for short) if $p$-almost every $\bar e$ is a best response to $p$.\footnote{Note that if $\bar e$ is a best response to $p$ then in particular $U(\bar e, p)\ge 0$, so staying out of the contest is not a profitable deviation. This follows from $\lim_{e\downarrow 0} U(e,p)\ge 0$.} We have the following.


\begin{proposition}\label{prop-symmetric-eq}
    Suppose that $W$ is an Additive RPF with associated cdf $F$, that $F$ admits a strictly positive and differentiable density $f$, and that $c''(e)>-f'(s)u(V)$ for all $e\in E$ and $s\in \RR$. Then the contest has a unique equilibrium $p=\delta_{e^*}$, where $e^*$ solves
    \begin{equation}\label{eqn-FOC}
    c'(e^*)=f(F^{-1}(1-k))u(V).
\end{equation}
\end{proposition}

\begin{proof}
Fix some $p\in \Delta_k$. We first argue that the best response to $p$ is unique. Indeed, we have that 
$$U(e,p)=W(e,p)u(V)-c(e) = (1-F(s(p)-e))u(V)-c(e),$$
so the condition in the proposition implies that $U(e,p)$ is strictly concave in $e$. Further, the first order condition is given by 
$$f(s(p)- e)u(V)-c'(e)=0.$$
Since $f$ is strictly positive, and since $\lim_{e\downarrow 0} c'(e)=0$, the difference $f(s(p)- e)u(V)-c'(e)$ is positive for all $e$ sufficiently close to zero; and this difference is clearly negative for all $e$ large enough. Thus, a best response exists and is unique.

It follows that if $p$ is an equilibrium then $p=\delta_{\bar e}$ for some $\bar e\in E$. By the market clearing condition (\ref{eqn-BB}) we have that $1-F(s(\delta_{\bar e})-\bar e)=k$, implying that $s(\delta_{\bar e})=F^{-1}(1-k)+\bar e$. The first order condition from above then becomes 
$$f(F^{-1}(1-k)+\bar e -e)u(V)-c'(e)=0.$$
For $\bar e$ to be a best response to $p=\delta_{\bar e}$ it is thus necessary and sufficient that 
$$f(F^{-1}(1-k))u(V)-c'(\bar e)=0.$$
This is condition (\ref{eqn-FOC}) from the proposition, and clearly there is a unique $\bar e$ that satisfies it.

\end{proof}

We note that the second order condition in the proposition holds for example if $f'$ is bounded and $c(e)=Ae^2$ for sufficiently large $A$. Of course, even without global concavity $e^*$ from (\ref{eqn-FOC}) may still be an equilibrium.

\subsection{Comparative statics}

Equation (\ref{eqn-FOC}) which characterizes equilibrium effort facilitates simple comparative statics with respect to the model parameters.\footnote{In what follows we assume that $e^*$ given by (\ref{eqn-FOC}) is indeed an equilibrium. This limits the scope of our conclusions.} Since $c'$ is assumed to be increasing, the direction of change of $e^*$ is the same as that of the right-hand side of (\ref{eqn-FOC}). Thus, equilibrium effort increases in the value of the prize $V$, and the dependence on the budget $k$ is determined by the behavior of the function $f(F^{-1}(\cdot))$. In particular, if $F$ is symmetric around zero and the density satisfies $sf'(s)\le 0$ for all $s\in \RR$, then $e^*$ is increasing for $k\in (0,0.5)$ and decreasing for $k\in (0.5,1)$, so that maximal effort is obtained at $k=0.5$. 

We can also predict how an increase in noise would affect equilibrium effort. \citet{drugov2020noise} study this question in the classic \citet{lazear1981rank} setup with a finite number of agents. They show that equilibrium effort is higher under $F$ than under $\tilde F$ in every contest (any number of players and any prize structure) if and only if $f(F^{-1}(q))\ge \tilde{f}(\tilde{F}^{-1}(q))$ for every quantile $q\in (0,1)$. From (\ref{eqn-FOC}) it immediately follows that a similar conclusion holds in our setup.\footnote{To provide intuition for their result, \citet[pages 2-3]{drugov2020noise} informally consider a large contest version of their model and arrive to the same equilibrium condition as in our equation (\ref{eqn-FOC}). We are grateful to Mikhail Drugov for pointing us to their related papers \citep{drugov2020noise, drugov2020tournament}.} Furthermore, from Corollary \ref{coro-compare} we know that, with log-concave densities,  $f[F^{-1}(\cdot)]\ge \tilde{f}[\tilde{F}^{-1}(\cdot)]$ is equivalent to $F$ better reflecting effort than $\tilde F$. This gives the connection between the contest success function induced by the noise structure and the level of equilibrium effort.

\subsection{The optimal prize structure}

In the contest just described, suppose that the principal's goal is to maximize the equilibrium effort $e^*$. The principal has total budget $B=1$ and can choose the fraction of winners $k$ and the size of each award $V$ subject to the constraint $kV=1$. \citet{krishna1998winner} study this question in the same framework as ours but with a small number of agents (at most four), and their analysis was later extended to any finite number of agents by \citet{drugov2020tournament}.\footnote{Other papers that study the optimal design of prizes in different types of contests include \citet{moldovanu2001optimal}, \citet{fang2020turning} and \citet{olszewski2020performance}, among many others.} An important result of these papers is that if $F$ is symmetric around zero then it is optimal to allocate the entire budget to agents in the top half of the performance ranking.\footnote{\citet{krishna1998winner} also require that the density $f$ is unimodal, but \citet{drugov2020tournament} show that this is not needed. Both papers allow for heterogeneous prizes, while we restrict the principal to give the same prize to all winners.} The same result holds in our setup with a continuum of agents.

\begin{proposition}\label{prop-optimal-k}
Assume that $F$ is symmetric around zero and let $e^*(k)$ be the solution to (\ref{eqn-FOC}) with $V=\frac{1}{k}$. Then $e^*(k)<e^*(1-k)$ for every $k\in (0.5,1)$, so allocating awards to more than half of the population is never optimal. 
\end{proposition}

\begin{proof}
For $k\in (0.5,1)$ we have $u\left(\frac{1}{k}\right)< u\left(\frac{1}{1-k}\right)$. Also, $f(F^{-1}(k))= f(F^{-1}(1-k))$ by the symmetry of $F$. Therefore,
$$f(F^{-1}(k))u\left(\frac{1}{k}\right) < f(F^{-1}(1-k))u\left(\frac{1}{1-k}\right),$$
so by (\ref{eqn-FOC}) we are done.
\end{proof}

When agents are risk neutral, finding the optimal $k$ boils down to finding the maximum of the hazard ratio $\frac{f}{1-F}$. Indeed, the right-hand side of (\ref{eqn-FOC}) in this case is given by $f(F^{-1}(1-k))\frac{1}{k}$, which after a change of variable $s=F^{-1}(1-k)$ becomes $\frac{f(s)}{1-F(s)}$. Thus, $e^*(k_1)>e^*(k_2)$ if and only if the hazard ratio $\frac{f}{1-F}$ is larger at the point $F^{-1}(1-k_1)$ than at the point $F^{-1}(1-k_2)$. In particular, if $F$ has an increasing hazard ratio, as is the case for the normal distribution, then $e^*(k)$ is decreasing in $k$ so reducing the fraction of winners increases equilibrium effort. But if the hazard ratio is maximized at some finite $s$, e.g.\ in the family of Student's $t$ distributions, then there would be an interior optimal $k\in (0,0.5)$. Figure \ref{figure-optimal-k} illustrates this for several distributions. \citet[Proposition 2]{drugov2020tournament} obtain a similar result with finitely many agents: When the hazard ratio is increasing, a winner-take-all tournament is optimal, and when it is decreasing it is optimal to spread the budget evenly among all players except the one who ranked last.

\begin{figure}
\centering
\includegraphics[width=.3\textwidth]{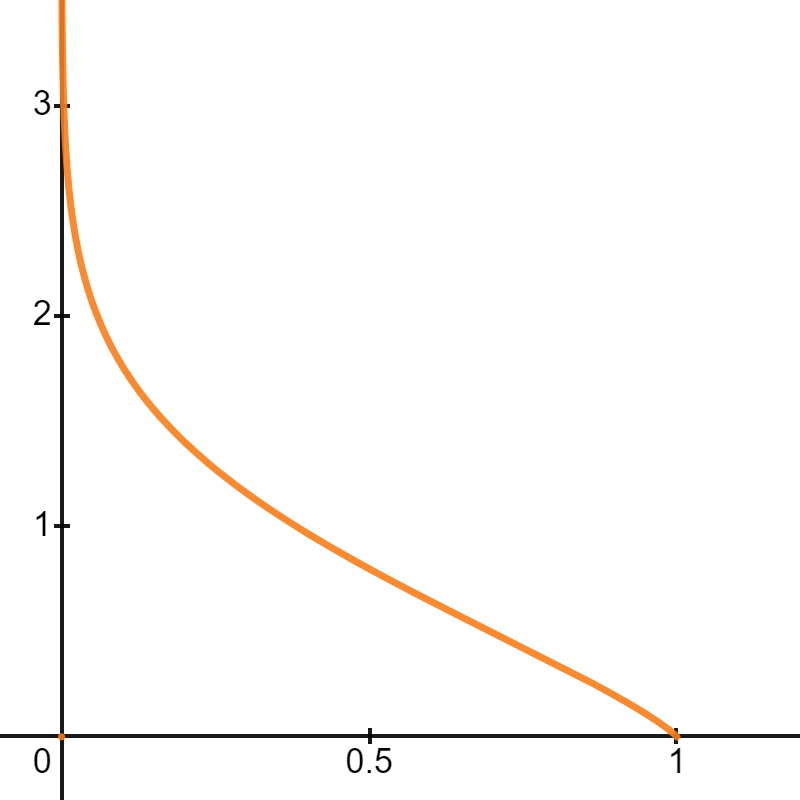}\hfill
\includegraphics[width=.3\textwidth]{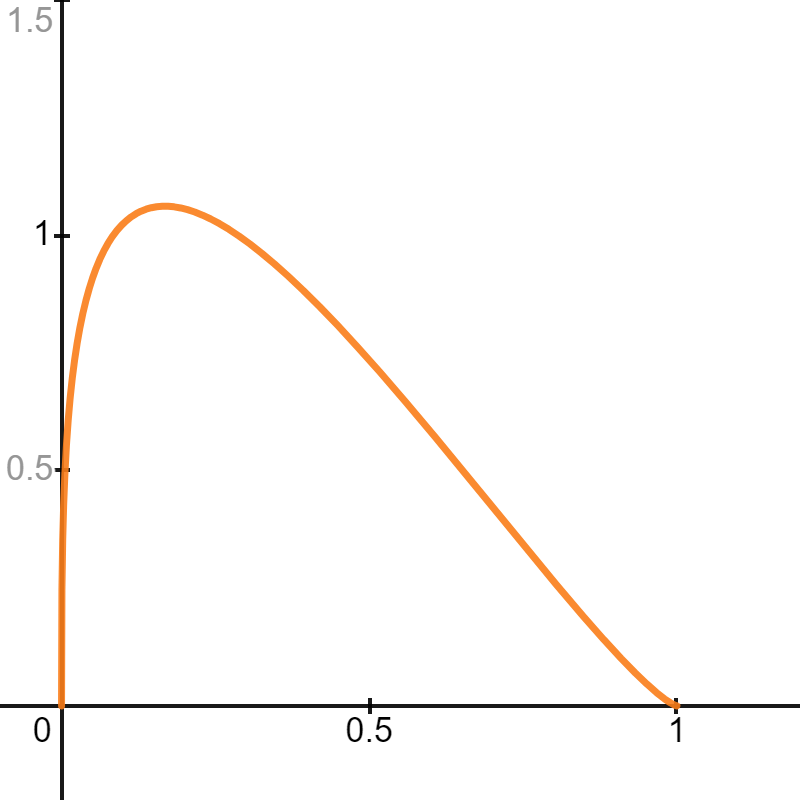}\hfill
\includegraphics[width=.3\textwidth]{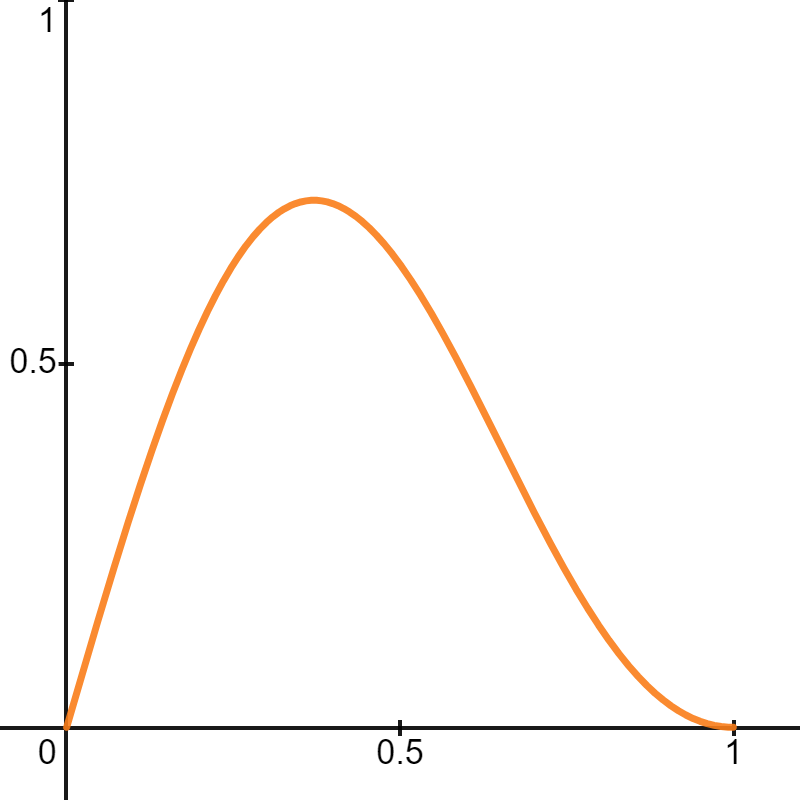}
\caption{The figure shows the graphs of $\frac{1}{k}f(F^{-1}(1-k))$ as a function of $k\in (0,1)$ for the cases of a standard normal distribution (left panel), Student's $t$ distribution with 3 degrees of freedom (middle panel), and Student's $t$ distribution with 1 degree of freedom (right panel). In the first case the graph decreases in $k$, implying that a lower fraction of winners yields higher equilibrium effort, while in the other two cases the maximum is attained at an interior point (with higher optimal $k$ corresponding to less degrees of freedom).}
\label{figure-optimal-k}
\end{figure}

\subsection{Rent dissipation}

A key issue in the contests literature has been the extent of rent dissipation -- the relationship between the rents obtained by winning and the total cost of effort exerted in order to capture them.\footnote{See for example \citet{nitzan1994modelling} for an early survey in the case of a single prize, and \citet{berry1993rent} and \citet{clark1996multi} for the case of multiple winners.} Let us return to the contest model described above with fixed $k$, and for expositional reasons suppose that agents are risk neutral, $u(V)=V$, and that the cost function is quadratic, $c(e)=Ae^2$. With these assumptions, the equilibrium effort $e^*$ from Proposition \ref{prop-symmetric-eq} is given by $e^*=\frac{V}{2A} f(F^{-1}(1-k))$, and therefore the equilibrium cost of effort is
$$c(e^*) = A(e^*)^2 = \frac{V^2}{4A} (f(F^{-1}(1-k)))^2.$$
Total rents in the contest are $kV$, and it is straightforward to calculate the dissipation of rents and how it is affected by the size of the award $V$, the budget $k$, and the noise distribution $F$. In particular, the ratio of costs to rents
$$ \frac{V^2}{4A} (f(F^{-1}(1-k)))^2 / kV = \frac{V}{4Ak} (f(F^{-1}(1-k)))^2 $$
is less than one for all sufficiently small $V$ and becomes larger than one when $V$ crosses a threshold. Thus, we can have either under-dissipation or over-dissipation of rents depending on the size of the award, unlike in the classic Tullock contest with linear costs where this ratio is independent of $V$.

\bibliographystyle{plainnat}
\bibliography{success_functions}

\begin{thebibliography}{31}
\providecommand{\natexlab}[1]{#1}
\providecommand{\url}[1]{\texttt{#1}}
\expandafter\ifx\csname urlstyle\endcsname\relax
  \providecommand{\doi}[1]{doi: #1}\else
  \providecommand{\doi}{doi: \begingroup \urlstyle{rm}\Url}\fi

\bibitem[Adda and Ottaviani(2023)]{adda2023grantmaking}
J{\'e}r{\^o}me Adda and Marco Ottaviani.
\newblock Grantmaking, grading on a curve, and the paradox of relative evaluation in nonmarkets.
\newblock \emph{The Quarterly Journal of Economics}, page qjad046, 2023.

\bibitem[Athey and Levin(2018)]{athey2018value}
Susan Athey and Jonathan Levin.
\newblock The value of information in monotone decision problems.
\newblock \emph{Research in Economics}, 72\penalty0 (1):\penalty0 101--116, 2018.

\bibitem[Azrieli(2024)]{azrieli2024temporary}
Yaron Azrieli.
\newblock Temporary exclusion in repeated contests.
\newblock \emph{arXiv preprint arXiv:2401.06257}, 2024.

\bibitem[Berry(1993)]{berry1993rent}
S~Keith Berry.
\newblock Rent-seeking with multiple winners.
\newblock \emph{Public Choice}, 77\penalty0 (2):\penalty0 437--443, 1993.

\bibitem[Clark and Riis(1996)]{clark1996multi}
Derek~J Clark and Christian Riis.
\newblock A multi-winner nested rent-seeking contest.
\newblock \emph{Public Choice}, 87:\penalty0 177--184, 1996.

\bibitem[Clark and Riis(1998)]{clark1998contest}
Derek~J Clark and Christian Riis.
\newblock Contest success functions: an extension.
\newblock \emph{Economic Theory}, 11:\penalty0 201--204, 1998.

\bibitem[Do{\u{g}}an et~al.(2023)Do{\u{g}}an, Karag{\"o}zo{\u{g}}lu, Keskin, and Sa{\u{g}}lam]{dougan2023large}
Serhat Do{\u{g}}an, Emin Karag{\"o}zo{\u{g}}lu, Kerim Keskin, and {\c{C}}a{\u{g}}r{\i} Sa{\u{g}}lam.
\newblock Large tullock contests.
\newblock \emph{Journal of Economics}, pages 1--11, 2023.

\bibitem[Drugov and Ryvkin(2020{\natexlab{a}})]{drugov2020noise}
Mikhail Drugov and Dmitry Ryvkin.
\newblock How noise affects effort in tournaments.
\newblock \emph{Journal of Economic Theory}, 188:\penalty0 105065, 2020{\natexlab{a}}.

\bibitem[Drugov and Ryvkin(2020{\natexlab{b}})]{drugov2020tournament}
Mikhail Drugov and Dmitry Ryvkin.
\newblock Tournament rewards and heavy tails.
\newblock \emph{Journal of Economic Theory}, 190:\penalty0 105116, 2020{\natexlab{b}}.

\bibitem[Fang et~al.(2020)Fang, Noe, and Strack]{fang2020turning}
Dawei Fang, Thomas Noe, and Philipp Strack.
\newblock Turning up the heat: The discouraging effect of competition in contests.
\newblock \emph{Journal of Political Economy}, 128\penalty0 (5):\penalty0 1940--1975, 2020.

\bibitem[Fu and Lu(2012)]{fu2012micro}
Qiang Fu and Jingfeng Lu.
\newblock Micro foundations of multi-prize lottery contests: a perspective of noisy performance ranking.
\newblock \emph{Social Choice and Welfare}, 38\penalty0 (3):\penalty0 497--517, 2012.

\bibitem[Krishna and Morgan(1998)]{krishna1998winner}
Vijay Krishna and John Morgan.
\newblock The winner-take-all principle in small tournaments.
\newblock \emph{Advances in applied microeconomics}, 7:\penalty0 61--74, 1998.

\bibitem[Lahkar and Mukherjee(2023)]{lahkar2023optimal}
Ratul Lahkar and Saptarshi Mukherjee.
\newblock Optimal large population tullock contests.
\newblock \emph{Oxford Open Economics}, page odad003, 2023.

\bibitem[Lahkar and Sultana(2023)]{lahkar2023rent}
Ratul Lahkar and Rezina Sultana.
\newblock Rent dissipation in large population tullock contests.
\newblock \emph{Public Choice}, 197\penalty0 (1):\penalty0 253--282, 2023.

\bibitem[Lazear and Rosen(1981)]{lazear1981rank}
Edward~P Lazear and Sherwin Rosen.
\newblock Rank-order tournaments as optimum labor contracts.
\newblock \emph{Journal of political Economy}, 89\penalty0 (5):\penalty0 841--864, 1981.

\bibitem[Lehmann(1988)]{lehmann1988comparing}
EL~Lehmann.
\newblock Comparing location experiments.
\newblock \emph{The Annals of Statistics}, pages 521--533, 1988.

\bibitem[Luce(1959)]{luce1959individual}
R~Duncan Luce.
\newblock \emph{Individual choice behavior: A theoretical analysis}.
\newblock John Wiley \& Sons, Inc., 1959.

\bibitem[Moldovanu and Sela(2001)]{moldovanu2001optimal}
Benny Moldovanu and Aner Sela.
\newblock The optimal allocation of prizes in contests.
\newblock \emph{American Economic Review}, 91\penalty0 (3):\penalty0 542--558, 2001.

\bibitem[Morgan et~al.(2018)Morgan, Sisak, and V{\'a}rdy]{morgan2018ponds}
John Morgan, Dana Sisak, and Felix V{\'a}rdy.
\newblock The ponds dilemma.
\newblock \emph{The Economic Journal}, 128\penalty0 (611):\penalty0 1634--1682, 2018.

\bibitem[M{\"u}nster(2009)]{munster2009group}
Johannes M{\"u}nster.
\newblock Group contest success functions.
\newblock \emph{Economic Theory}, 41\penalty0 (2):\penalty0 345--357, 2009.

\bibitem[Myerson and W{\"a}rneryd(2006)]{myerson2006population}
Roger~B Myerson and Karl W{\"a}rneryd.
\newblock Population uncertainty in contests.
\newblock \emph{Economic Theory}, 27:\penalty0 469--474, 2006.

\bibitem[Nitzan(1994)]{nitzan1994modelling}
Shmuel Nitzan.
\newblock Modelling rent-seeking contests.
\newblock \emph{European Journal of Political Economy}, 10\penalty0 (1):\penalty0 41--60, 1994.

\bibitem[Olszewski and Siegel(2016)]{olszewski2016large}
Wojciech Olszewski and Ron Siegel.
\newblock Large contests.
\newblock \emph{Econometrica}, 84\penalty0 (2):\penalty0 835--854, 2016.

\bibitem[Olszewski and Siegel(2019)]{olszewski2019bid}
Wojciech Olszewski and Ron Siegel.
\newblock Bid caps in large contests.
\newblock \emph{Games and Economic Behavior}, 115:\penalty0 101--112, 2019.

\bibitem[Olszewski and Siegel(2020)]{olszewski2020performance}
Wojciech Olszewski and Ron Siegel.
\newblock Performance-maximizing large contests.
\newblock \emph{Theoretical Economics}, 15\penalty0 (1):\penalty0 57--88, 2020.

\bibitem[Quah and Strulovici(2009)]{quah2009comparative}
John K-H Quah and Bruno Strulovici.
\newblock Comparative statics, informativeness, and the interval dominance order.
\newblock \emph{Econometrica}, 77\penalty0 (6):\penalty0 1949--1992, 2009.

\bibitem[Rai and Sarin(2009)]{rai2009generalized}
Birendra~K Rai and Rajiv Sarin.
\newblock Generalized contest success functions.
\newblock \emph{Economic Theory}, 40:\penalty0 139--149, 2009.

\bibitem[Sisak(2009)]{sisak2009multiple}
Dana Sisak.
\newblock Multiple-prize contests--the optimal allocation of prizes.
\newblock \emph{Journal of Economic Surveys}, 23\penalty0 (1):\penalty0 82--114, 2009.

\bibitem[Skaperdas(1996)]{skaperdas1996contest}
Stergios Skaperdas.
\newblock Contest success functions.
\newblock \emph{Economic Theory}, 7:\penalty0 283--290, 1996.

\bibitem[Tullock(1980)]{tullock1980efficient}
Gordon Tullock.
\newblock Efficient rent seeking. jm buchanan, rd tollison, g. tullock, eds., toward a theory of the rent seeking society.
\newblock \emph{A \& M University Press, College Station, TX: Texas}, 1980.

\bibitem[Tullock(2001)]{tullock2001efficient}
Gordon Tullock.
\newblock Efficient rent seeking.
\newblock In \emph{Efficient rent-seeking: Chronicle of an intellectual quagmire}, pages 3--16. Springer, 2001.

\end{thebibliography}





\end{document}